\documentclass[a4paper,11pt]{article}

\usepackage[latin1]{inputenc}
\usepackage{amsmath,amssymb}
\usepackage{amsthm}
\usepackage{a4wide}
\usepackage{graphicx}
\usepackage{enumerate}
\usepackage{float}
\usepackage{calc}

\newtheorem{theorem}{Theorem}[section]
\newtheorem{lemma}[theorem]{Lemma}

\newtheorem{corollary}[theorem]{Corollary}
\newtheorem{proposition}[theorem]{Proposition}

\floatstyle{ruled}
\newfloat{algorithm}{htp!}{loa}
\floatname{algorithm}{Algorithm}

\newlength{\algindent}
\newenvironment{AlgorithmSteps}{%
  \setlength{\algindent}{0mm}
  \small
  \begin{enumerate}[\bf\small 1.]%
  \setlength{\itemsep}{3pt} 
}{%
  \end{enumerate}
}
\newcommand{\Step}[1]{\item\hspace*{\algindent}\parbox[t]{\textwidth-\algindent-\leftmargin}{#1}}
\newcommand{\Comment}[1]{\item[]\hspace*{\algindent}\parbox[t]{\textwidth-\algindent-\leftmargin}{\sl\{#1\}}}
\newcommand{\IncreaseIndent}{\addtolength{\algindent}{7mm}}
\newcommand{\DecreaseIndent}{\addtolength{\algindent}{-7mm}}

\newcommand{\MAX}{\mathrm{MAX}}
\newcommand{\MIN}{\mathrm{MIN}}

\newenvironment{AlgorithmBlock}{\IncreaseIndent}{\DecreaseIndent}

\newcommand{\Input}[1]{\par\noindent\textbf{Input:} #1}
\newcommand{\Output}[1]{\par\noindent\textbf{Output:} #1}

\newcommand{\MS}{\ensuremath{\mathcal{S}}}

\begin{document}

  \title{Hereditary biclique-Helly graphs: recognition and maximal biclique enumeration}

  \author{%
    Martiniano Egu\'\i a \and Francisco J.\ Soulignac
  }

  \date{\normalsize Universidad de Buenos Aires, Facultad de Ciencias Exactas y Naturales, Departamento  de  Computaci\'on, Buenos Aires, Argentina.\\\texttt{\{meguia, fsoulign\}@dc.uba.ar}}

\maketitle   

\begin{abstract}
A biclique is a set of vertices that induce a bipartite complete graph.  A graph $G$ is biclique-Helly when its family of maximal bicliques satisfies the Helly property.  If every induced subgraph of $G$ is also biclique-Helly, then $G$ is hereditary biclique-Helly.  A graph is $C_4$-dominated when every cycle of length $4$ contains a vertex that is dominated by the vertex of the cycle that is not adjacent to it.  In this paper we show that the class of hereditary biclique-Helly graphs is formed precisely by those $C_4$-dominated graphs that contain no triangles and no induced cycles of length either $5$, or $6$.  Using this characterization, we develop an algorithm for recognizing hereditary biclique-Helly graphs in $O(n^2+\alpha m)$ time and $O(m)$ space.  (Here $n$, $m$, and $\alpha = O(m^{1/2})$ are the number of vertices and edges, and the arboricity of the graph, respectively.)  As a subprocedure, we show how to recognize those $C_4$-dominated graphs that contain no triangles in $O(\alpha m)$ time and $O(m)$ space.  Finally, we show how to enumerate all the maximal bicliques of a $C_4$-dominated graph with no triangles in $O(n^2 + \alpha m)$ time and $O(\alpha m)$ space, and we discuss how some biclique problems can be solved in $O(\alpha m)$ time and $O(n+m)$ space.  

\vspace*{3mm} {\bf Keywords:} hereditary biclique-Helly graphs, maximal bicliques, triangle-free graphs, domination problems.

\end{abstract}

\section{Introduction}

A famous theorem by Helly states that, in a $d$-dimensional euclidean space, if in a finite collection of $n > d$ convex sets any $d + 1$ sets have a point in common, then there is a point in common to all the sets~\cite{HellyJDMV1923}.  The Helly property generalizes this theorem for families of sets of any kind.  A family of sets satisfies the Helly property, or simply is a Helly family, if for every subfamily $\mathcal{F}$ of pairwise intersecting sets there is an element common to all the sets in $\mathcal{F}$.

The Helly property arises naturally in the graph theory field~\cite{BrandstadtLeSpinrad1999,Golumbic2004,McKeeMcMorris1999}.  In the study of clique graphs, the Helly property plays a central role.  Roberts and Spencer proved that a graph is a clique graph if and only if there is a Helly family of cliques that covers all the edges of the graph~\cite{RobertsSpencerJCTSB1971} (see also~\cite{HamelinkJCT1968}).  Based on this result, it is interesting to study the subclass of clique graphs in which the family of maximal cliques is Helly.  Such is the class of clique-Helly graphs.  Szwarcfiter presented a characterization of clique-Helly graphs that yields a polynomial time algorithm for the associated recognition problem~\cite{SzwarcfiterAC1997}, while Lin and Szwarcfiter recently developed an $O(m^2)$ time recognition algorithm~\cite{LinSzwarcfiterIPL2007}.  On the other hand, Alcón et al.\ showed that the recognition of clique graphs is an NP-complete problem~\cite{Alc'onFariaFigueiredoGutierrezTCS2009}.  The Helly property has been applied in a similar fashion to other families of vertex sets so as to define several other classes of graphs. A survey on the Helly property for graphs, mainly from a complexity point of view, is given in~\cite{DouradoProttiSzwarcfiterEJC2009}.

A clique-Helly graph can be obtained from any given graph, by inserting a new vertex adjacent to all the existing vertices.  As a consequence, the class of clique-Helly graphs is not closed under induced subgraphs.  It makes sense to study those graphs whose all their induced subgraphs are clique-Helly.  These graphs are known as the hereditary clique-Helly graphs.  Prisner gave a characterization of hereditary clique-Helly graphs by means of four forbidden induced subgraphs with $6$ vertices each.  He also showed an $O(n^2m)$ time algorithm for the associated recognition problem~\cite{PrisnerJCMCC1993}.  Later, Lin and Szwarcfiter developed an improved $O(m^2)$ time and $O(\alpha m)$ space recognition algorithm~\cite{LinSzwarcfiterIPL2007}, where $\alpha < \sqrt{m}$ is the arboricity of the graph.  

The problem of enumerating all the maximal cliques of a graph is widely studied, both for the general case (e.g.~\cite{ChibaNishizekiSJC1985,JohnsonYannakakisPapadimitriouIPL1988,TsukiyamaIdeAriyoshiShirakawaSJC1977,TomitaTanakaTakahashiTCS2006}) and for some restricted graph classes (see e.g.~\cite{BrandstadtLeSpinrad1999,Golumbic2004}).  Since a clique-Helly graph can be obtained from any graph by inserting a universal vertex, any algorithm for the enumeration of the maximal cliques of a clique-Helly graph can also be used to enumerate all the maximal cliques of a general graph.  That is, the best algorithms for enumerating the maximal cliques of clique-Helly graphs are of general purpose.  For hereditary clique-Helly graphs the situation is quite different, since clique-Helly graphs have at most $m$ maximal cliques~\cite{PrisnerJCMCC1993}.  This $O(m)$ bound follows from the fact that all the maximal cliques of a hereditary clique-Helly graph must have an edge that belongs to no other maximal clique~\cite{PrisnerJCMCC1993,WallisZhangJCMCC1990}.  So, using the same ideas as in~\cite{LinSzwarcfiterIPL2007}, an $O(m^2)$ time algorithm for enumerating all the maximal cliques yields from this property.  

In this paper we focus our attention on bicliques.  A biclique is a set of vertices inducing a bipartite complete subgraph.  The problem of enumerating all the (non-induced) maximal bicliques of a general graph is also widely studied (e.g.~\cite{AlexeAlexeCramaFoldesHammerSimeoneDAM2004,Binkele-RaibleFernauGaspersLiedloffIPL2010,DiasFigueiredoSzwarcfiterTCS2005,DiasFigueiredoSzwarcfiterDAM2007,G'elyNourineSadiDAM2009}), as well as it is the problem of generating the maximal bicliques of a bipartite graph~\cite{MakinoUno2004,NourineRaynaudIPL1999}.  In the last years, other concepts that are widely studied for cliques have been studied in terms of bicliques~\cite{Groshaus2006,GroshausMontero2008,GroshausMontero2009,GroshausSzwarcfiterDMTCS2008,GroshausSzwarcfiterGC2007,GroshausSzwarcfiterJGT2010,Montero2008,Terlisky2010}. Groshaus and Szwarcfiter found a characterization of biclique graphs that somehow resembles the characterization of clique graphs by Roberts and Spencer~\cite{GroshausSzwarcfiterJGT2010}; in this case, a key ingredient of the characterization is a variation of the Helly property, which the authors call the bipartite Helly property.  Groshaus and Szwarcfiter also provided a characterization of biclique-Helly graphs, that is somehow related to the characterization given by Szwarcfiter for clique-Helly graphs, that also leads to a polynomial time algorithm for the associated recognition problem~\cite{GroshausSzwarcfiterGC2007}.  

As clique-Helly graphs, the induced subgraph of a biclique-Helly graph needs not be a biclique-Helly graph.  Indeed, by inserting a vertex adjacent to all the vertices in one bipartition of a bipartite graph, a biclique-Helly graph is obtained.  A graph is hereditary biclique-Helly graphs if all its induced subgraphs are biclique-Helly.  Groshaus and Szwarcfiter also studied the class of hereditary biclique-Helly graphs~\cite{GroshausSzwarcfiterDMTCS2008}.  They showed a family of six forbidden induced subgraphs with at most $8$ vertices.  As a corollary, the recognition of hereditary biclique-Helly graphs takes polynomial time, though the most efficient algorithm to this date takes $O(n^3m^2)$ time (cf.~\cite{DouradoProttiSzwarcfiterEJC2009}).

Prisner proved that bipartite graphs can have an exponential number of bicliques~\cite{PrisnerC2000}.  Since a biclique-Helly graph can be obtained from any bipartite graph by the insertion of one vertex, biclique-Helly graphs can have an exponential number of maximal bicliques.  Furthermore, as for clique-Helly graphs, the best algorithms for listing all the maximal bicliques of a biclique-Helly graph are of general purpose.  

In this paper, we consider two problems related with hereditary biclique-Helly graphs: their recognition and the enumeration of their maximal bicliques.  For the recognition problem, we rephrase the characterization by Groshaus and Szwarcfiter in more algorithmic terms and, using this new characterization, we develop an $O(\alpha m + n^2)$ time and $O(m)$ space recognition algorithm.  As one of the steps in our algorithm, we require the recognition of a larger class, formed by all the triangle-free graphs whose $C_4$'s have at least one vertex dominated by other vertex of the cycle.  We call this class, the class of $C_4$-dominated graphs with no triangles.  We develop a recognition algorithm for this class that takes $O(\alpha m)$ time and $O(m)$ space. For the enumeration problem, we develop an $O(\alpha m+o)$ time and space algorithm that outputs all the bicliques of any triangle-free $C_4$-dominated graph, where $o < n^2$ is the size of the output.  Furthermore, we prove that every maximal biclique of a graph in this class is formed by those vertices that either are adjacent or dominate $v$, for some vertex $v$.  As a result, hereditary biclique-Helly graphs can have at most $n$ maximal bicliques.

The article is organized as follows.  In the next section we introduce the notation and terminology employed. In Section~\ref{sec:simple recognition} we develop a simple $O(nm)$ time and $O(n^2)$ space algorithm for the recognition of biclique-Helly graphs.  Following, in Sections \ref{sec:C4 fast}~and~\ref{sec:C6 fast}, we present an improved implementation of this simple algorithm, so that it runs in $O(\alpha m + n^2)$ time and $O(m)$ space.  The algorithm for enumerating the maximal bicliques of $C_4$-dominated graph with no triangles is given in Section~\ref{sec:bicliques}. Finally, in Section~\ref{sec:remarks}, we give some remarks and leave some open problems.

\section{Preliminaries}
\label{sec:preliminaries}

In this paper we work with simple graphs.  Let $G$ be a graph with vertex set $V(G)$ and edge set $E(G)$, and call $n = |V(G)|$ and $m = |E(G)|$.  Write $vw$ to denote the edge of $G$ formed by vertices $v, w \in V(G)$.  For $v \in V(G)$, represent by $N_G(v)$ the set of vertices adjacent to $v$. The set $N_G(v)$ is called the \emph{neighborhood} of $v$, and $d_G(v) = |N(v)|$ is the \emph{degree of $v$}.  Similarly, for $v,w \in V(G)$, define the \emph{common neighborhood} of $v$ and $w$ as $N_G(vw) = N_G(v) \cap N_G(w)$.  Say that $v$ \emph{dominates} $w$, or equivalently that $w$ is \emph{dominated} by $v$, when $N(w) \subseteq N(v)$.  Note that $v$ dominates $w$ only if $v$ is not adjacent to $w$.  The set of vertices that dominate $v$ is represented by $Dom_G(v)$.  Say that $v$ is \emph{dom-comparable} to $w$ when $v$ either dominates or is dominated by $w$.  When there is no ambiguity, we may omit the subscripts from $N$, $Dom$, and $d$.

Say that a total ordering $<$ of $V(G)$ is a \emph{degree ordering} when $v < w$ only if $d(v) \leq d(w)$.  A \emph{degree ordered} graph is a pair $(G, <)$, where $G$ is a graph and $<$ is a degree ordering of $G$.  For the sake of simplicity, we say that $G$ is a \emph{degree ordered graph} to indicate that there is a degree ordering $<$ such that $(G, <)$ is a degree ordered graph. For a vertex $v$ of a degree ordered graph $G$, define $\MAX_G(v) = \{w \in V(G) \mid w > v\}$ and $\MIN_G(v) = \{w \in V(G) \mid w < v\}$.  For $W \subseteq V(G)$, we also use $\max_G W$ and $\min_G W$ to refer to the maximum and minimum elements of $W$, according to $<$.  As before, we omit the subscript from $\MAX$, $\MIN$, $\max$, and $\min$ when there is no ambiguity.  For two vertices $v > w$, define the \emph{least common neighborhood} as $L(v,w) = N(vw) \cap \MIN(v)$; note that, by definition, $L(w,v)$ is undefined for $v > w$.

For $W \subseteq V(G)$, denote by $G[W]$ the subgraph of $G$ induced by $W$.  An \emph{independent set} is a set $W \subseteq V(G)$ formed by pairwise non-adjacent vertices. Graph $G$ is \emph{bipartite} when $V(G)$ can be partitioned into two independent sets $W_1$ and $W_2$, where possibly $W_2 = \emptyset$. In this case, the unordered pair $\{W_1, W_2\}$ is called a \emph{bipartition} of $G$. Furthermore, if $vw$ is an edge of $G$ for every $v \in W_1$ and $w \in W_2$, then $G$ is a \emph{bipartite complete} graph. A \emph{biclique} $H$ of $G$ is a bipartite complete induced subgraph of $G$; we also use the term \emph{biclique} to refer to both $V(H)$ and the unique bipartition of $H$.

Let $\mathcal{F}$ be a family of sets.  Say that $\mathcal{F}$ is \emph{pairwise intersecting} when $S \cap T \neq \emptyset$, for every $S, T \in \mathcal{F}$, while $\mathcal{F}$ is \emph{globally intersecting} when $\bigcap\mathcal{F} \neq \emptyset$.  Family $\mathcal{F}$ is \emph{Helly} when all its pairwise intersecting subfamilies are globally intersecting.  A graph $G$ is \emph{biclique-Helly} when its family of maximal bicliques is Helly, and it is \emph{hereditary biclique-Helly} when all its induced subgraphs are biclique-Helly.  

Denote by $C_n$ the cycle graph with $n$ vertices; $C_3$ is also called a \emph{triangle}.  For a graph $H$, say that $G$ is $H$-free when no induced subgraph of $G$ is isomorphic to $H$.  Similarly, for a family $\mathcal{H}$ of graphs, say that $G$ is $\mathcal{H}$-free when $G$ is $H$-free for every $H \in \mathcal{H}$.  The \emph{arboricity} $\alpha(G)$ of $G$ is the minimum number of edge-disjoint spanning forests into which $G$ can be decomposed.  Chiba and Nishizeki proved that $\alpha(G) \leq m^{1/2}$~\cite{ChibaNishizekiSJC1985}.

In this paper we also work with simple digraphs.  Let $D$ be a graph with vertex set $V(D)$ and edge set $E(D)$.  Write $v \to_D w$ to indicate that the ordered pair $(v,w)$ is an edge of $D$.  When $(v, w)$ is not an edge of $D$, we write $v \not\to_D w$.  For $v \in V(D)$, define $N^+_D(v) = \{w \in V(D) \mid v \to w\}$ and $N^-_D(v) = \{w \in V(D) \mid w \to v\}$.  Sets $N^+_D(v)$ and $N^-_D(v)$ are respectively the \emph{out-neighborhood} and \emph{in-neighborhood} of $v$, while the members of $N^+_D(v)$ and $N^-_D(v)$ are the \emph{out-neighbors} and \emph{in-neighbors} of $v$, respectively.  The \emph{out-degree} and \emph{in-degree} are the values $d^+_D(v) = |N^+_D(v)|$ and $d^-_D(v) = |N^-_D(v)|$, respectively.  When there is no ambiguity, we may omit the subscripts from $\to$, $\not\to$, $N^+$, $N^-$, $d^+$, and $d^-$.

Groshaus and Szwarcfiter formulated the following characterization of hereditary biclique-Helly graphs, by means of minimal forbidden induced subgraphs.

\begin{theorem}[\cite{GroshausSzwarcfiterDMTCS2008}]\label{thm:GroshausSzwarcfiterDMTCS2008}
  A graph is hereditary biclique-Helly if and only if it is does not contain any triangles, $C_5$'s, $C_6$'s, nor ladders as induced subgraphs (see Figure~\ref{fig:ladders}).
\end{theorem}

\begin{figure}
  \hspace{\stretch{1}} \includegraphics{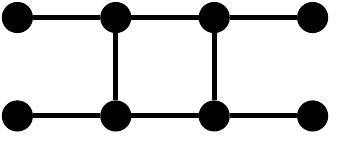} \hspace*{1cm} \includegraphics{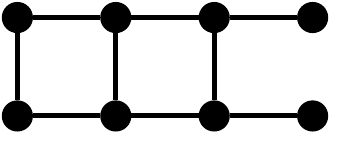} \hspace*{1cm} \includegraphics{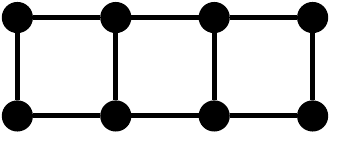} \hspace{\stretch{1}}
  \caption{The ladder graphs.}\label{fig:ladders}
\end{figure}

As a consequence of this theorem, the authors obtain an $O(n^3m^2)$ time algorithm for the recognition of hereditary biclique-Helly graphs (cf.~\cite{DouradoProttiSzwarcfiterEJC2009}).  

\section{Simple recognition of hereditary biclique-Helly graphs}
\label{sec:simple recognition}

In this section we rephrase Theorem~\ref{thm:GroshausSzwarcfiterDMTCS2008} in such a way that an $O(nm)$ time and $O(n^2)$ space recognition algorithm can be obtained with not to much effort.  To describe our algorithm, we require the following definitions for a graph $G$.  Say that a cycle of $G$ is \emph{dominated} if it contains a pair of dom-comparable vertices.  When every $C_4$ of $G$ is dominated, we say that $G$ is \emph{$C_4$-dominated}.  Theorem~\ref{thm:GroshausSzwarcfiterDMTCS2008} is rephrased as follows.

\begin{theorem}\label{thm:HBH recognition}
  A graph is hereditary biclique-Helly if and only if it is $C_4$-dominated and \{triangle, $C_5$, $C_6$\}-free.
\end{theorem}

\begin{proof}
  By Theorem~\ref{thm:GroshausSzwarcfiterDMTCS2008}, graphs that contain triangles, $C_5$'s, or $C_6$'s as induced subgraphs are not hereditary biclique-Helly.  Suppose now that $G$ is a \{triangle, $C_5$, $C_6$\}-free graph that contains a cycle $v_1, v_2, v_3, v_4$ in which neither $v_1$ and $v_3$ nor $v_2$ and $v_4$ are dom-comparable.  For the sake of notation, call $v_{i+4} = v_i$ for every $i \in \mathbb{Z}$.  Then, for $i \in \mathbb{Z}$, there is a vertex $w_i \in N(v_i) \setminus N(v_{i+2})$. (Again, call $w_{j+4} = w_j$ for every $j \in \mathbb{Z}$.) Since $G$ is triangle-free, $w_i$ is adjacent to neither $v_{i-1}$ nor $v_{i+1}$.  Hence, $w_i \neq w_j$ for every $1 \leq i \leq j \neq 4$.

  Clearly, $v_1, v_2, v_3, v_4$ is an induced cycle because $G$ is triangle-free.  Now, consider all the possible edges between the vertices in $\{w_1, w_2, w_3, w_4\}$.  First observe that $w_i$ is not adjacent to $w_{i+2}$; otherwise $w_iv_iv_{i+1}v_{i+2}w_{i+2}$ would induce a $C_5$ in $G$.  Similarly, $w_i$ is adjacent to neither $w_{i-1}$ nor $w_{i+1}$; otherwise $w_{i-1}v_{i-1}v_{i-2}v_{i+1}w_{i+1}w_{i}$ would induce a $C_6$ in $G$.  Therefore, the subgraph of $G$ induced by $\{v_iw_i\}_{1 \leq i \leq 4}$ isomorphic to a ladder and, by Theorem~\ref{thm:GroshausSzwarcfiterDMTCS2008}, $G$ is not hereditary biclique-Helly .
    
  For the converse, just observe that the cycle of a ladder formed by the vertices of degree at least $3$ is not dominated.  Then, the result follows from Theorem~\ref{thm:GroshausSzwarcfiterDMTCS2008}.
\end{proof}

Theorem~\ref{thm:HBH recognition} yields a simple three-step algorithm for the recognition of hereditary biclique-Helly graphs, summarized in Algorithm~\ref{alg:HBH recognition}.  Discuss its implementation.  For Step~\ref{alg:HBH recognition:triangle}, the algorithm in~\cite{ChibaNishizekiSJC1985} is called so as to find a triangle in $O(m\alpha(G))$ time and $O(m)$ space, when one exists.  In the rest of this section we discuss a simple $O(nm)$ time and $O(m^2)$ space implementation for Steps \ref{alg:HBH recognition:C4}~to~\ref{alg:HBH recognition:C6}.

\begin{algorithm}
  \caption{Recognition of hereditary biclique-Helly graphs.}\label{alg:HBH recognition}
  \Input{a graph $G$.}
  \Output{if $G$ is not hereditary biclique-Helly, then either a triangle, a non-dominated $C_4$, an induced $C_5$, or an induced $C_6$; otherwise, a message.}

  \begin{AlgorithmSteps}
    \Step{If $G$ contains a triangle $T$, then output $T$ and halt.}\label{alg:HBH recognition:triangle}
    \Step{If $G$ contains a non-dominated $C_4$ called $C$, then output $C$ and halt.}\label{alg:HBH recognition:C4}
    \Step{If $G$ contains an induced $C_5$ called $C$, then output $C$ and halt.}\label{alg:HBH recognition:C5}
    \Step{If $G$ contains an induced $C_6$ called $C$, then output $C$ and halt.}\label{alg:HBH recognition:C6}
    \Step{Output ``$G$ is HBH''.}
  \end{AlgorithmSteps}
\end{algorithm}

\subsection{An $O(nm)$ time implementation of Step~\ref{alg:HBH recognition:C4}}
\label{sec:slow}

The main tools for this step are the squares families.  Fix a degree ordered graph $G$ with no induced triangles for the rest of this section.  

For a vertex $v$, the \emph{squares family} of $v$ is the family $\MS(v)$ that contains one triple $S = (v, w, L(v,w))$ for each $w < v$ such that $L(v,w) \neq \emptyset$.  Refer to $v, w$ and $L(v,w)$ as the \emph{high vertex}, \emph{low vertex}, and \emph{common neighborhood} of $S$, respectively.  When $|L(v,w)| > 1$, the triple $S$ encodes all the $C_4$'s that contain $v$ and $w$, where $v$ is the maximum vertex of the $C_4$.  Indeed, $v,a,w,b$ is a $C_4$ of $G$ and $v > \max\{a,b,w\}$ if and only if $a, b \in L(v,w)$.  In this case, we say that $S$ \emph{represents} the cycle $v,a,w,b$, for every $a,b \in L(v,w)$.  The \emph{squares family} of $G$ is $\MS(G) = \bigcup_{v \in V(G)}{\MS(v)}$.  When $G$ is understood, we will simply write $\MS$ to mean $\MS(G)$.  Observe that every $C_4$ of $G$ is represented by exactly one triple of $\MS$.  Thus, the squares family of $G$ encodes of all the $C_4$'s of $G$ in $O(m\alpha(G))$ space, though $G$ could have $O(n^2)$ $C_4$'s~\cite{ChibaNishizekiSJC1985}.  For the sake of notation, write $v(S)$, $w(S)$, and $L(S)$ to respectively mean the high vertex, the low vertex, and the common neighborhood of $S$, for every $S \in \MS$.  Also, we sometimes write $(v,w)$ instead of $(v,w,L(v,w))$; we may write, for instance, that $(v,w) \in \MS$ to indicate that $(v,w,L(v,w)) \in \MS$.

Say that $S \in \MS$ is \emph{dominated} when all the $C_4$'s represented by $S$ are dominated.  (If $|L(S)| = 1$, then $S$ is vacuously dominated.)  By definition, $G$ is $C_4$-dominated if and only if $S$ is dominated, for every $S \in \MS$.  Say also that $S$ is \emph{safe} if $v(S)$ dominates $w(S)$, and that it is \emph{unsafe} otherwise.  Observe that if $S$ is safe, then it is also dominated.  Otherwise, $S$ is dominated if and only if $a$ dominates $b$, for every $a,b \in L(S)$ such that $a > b$.  These observations yield a simple algorithm to find a non-dominated $C_4$ of $G$, when one such $C_4$ exists, summarized as Algorithm~\ref{alg:C4 slow}.  

\begin{algorithm}
  \caption{Non-dominated $C_4$ of a triangle-free graph $G$.}\label{alg:C4 slow}
  \Input{a degree ordered graph $G$ with no induced triangles.}
  \Output{if existing, a non-dominated $C_4$ of $G$; otherwise, a message.}

  \begin{AlgorithmSteps}
    \Step{Let $v_1 > \ldots > v_n$ be the vertices of $G$.}\label{alg:C4 slow:ordering}
    \Step{Compute the matrix $D \in \{0,1\}^{n \times n}$ such that $d_{i,j} = 1$ if and only if $v_i$ dominates $v_j$, for $1 \leq i < j \leq n$.  Write $D(v_i, v_j) = d_{i,j}$.}\label{alg:C4 slow:matrix}
    \Step{For $i = 1$ to $n$, do:}
    \begin{AlgorithmBlock}    
      \Step{Compute $UNSAFE := \{(v_i, w, L(v_i,w)) \mid L(v_i,w) \neq \emptyset \text{ and } D(v_i, w) = 0\}$.}\label{alg:C4 slow:unsafe}
      \Step{For each $S \in UNSAFE$ do:}
      \begin{AlgorithmBlock}    
        \Step{Let $a_1 >  \ldots > a_{|L(S)|}$ be the vertices of $L(S)$}\label{alg:C4 slow:L ordering}
        \Step{If $D(a_j, a_{j+1}) = 0$ for $j \in \{1, \ldots, |L(S)|-1\}$, then output $v(S),a_j,w(S),a_{j+1}$ and halt.}\label{alg:C4 slow:L check}
      \end{AlgorithmBlock}
    \end{AlgorithmBlock}
    \Step{Output ``$G$ is $C_4$-dominated''.}
  \end{AlgorithmSteps}
\end{algorithm}

In Step~\ref{alg:C4 slow:unsafe}, Algorithm~\ref{alg:C4 slow} finds each unsafe $S \in \MS(v_i)$.  Following, the inner cycle checks that every such unsafe triple $S$ is dominated.  For this, it first computes a degree ordering $a_1 > \ldots > a_{|L(S)|}$ of $L(S)$, and then checks that $a_j$ dominates $a_{j+1}$, for every $1 \leq j < |L(S)|$.  If this check is fulfilled, then, since domination is a transitive relation, we obtain that $a$ dominates $b$ for every $a,b \in L(S)$ such that $a > b$.  Thus, Algorithm~\ref{alg:C4 slow} is correct.

Discuss the time complexity of the algorithm.  The matrix $D$ at Step~\ref{alg:C4 slow:matrix} can be obtained in $O(nm)$ time easily.  For Step~\ref{alg:C4 slow:unsafe}, we run one iteration of the method $C4$ developed by Chiba and Nishizeki in~\cite{ChibaNishizekiSJC1985}.  Each iteration of the method $C4$ takes $G$ and a vertex $v_i \in V(G)$ as input, and it outputs $\MS(v_i)$ in $O(\sum_{w > v_i}{d(w)}) = O(m)$ time and space.  (In fact, method $C4$ discards those $S \in \MS(v_i)$ for which $|L(S)| = 1$.  However, the algorithm can be easily modified so as to output these triples as well.) Furthermore, for each $S \in \MS(v_i)$, the list $L(S)$ given by the $C4$ is ordered in such a way that $a \in L(S)$ appears before $b \in L(S)$ if and only if $a > b$. (Here we assume that $a$ also appears before $b$ in $N(v_i)$.  By preprocessing $G$, such an ordering can be obtained in $O(n+m)$ time for every $v_i \in V(G)$.)  Thus, Step~\ref{alg:C4 slow:L ordering} is not actually executed.  Next, each $S \in \MS(v_i)$ is traversed so as to evaluate if it belongs to $UNSAFE$ in Step~\ref{alg:C4 slow:unsafe}.  For each $S \in UNSAFE$, the dominations at Step~\ref{alg:C4 slow:L check} are checked.  As each access to $D$ takes $O(1)$ time, each iteration of the outer loop takes $O(|L(S)|) = O(n)$ time.  Summing up, Algorithm~\ref{alg:C4 slow} takes $O(nm)$ time.  With respect to the space, the heaviest data structure used by the algorithm is the domination matrix $D$.  Thus, Algorithm~\ref{alg:C4 slow} requires $O(n^2)$ bits.

\subsection{An $O(nm)$ time implementation of Step~\ref{alg:HBH recognition:C5}}

For this paragraph, let again $G$ be a triangle-free graph.  Observe that every $C_5$ of $G$ must be induced; i.e., $G$ contains no induced $C_5$'s if and only if it contains no $C_5$'s at all.  It is not so hard to find a $C_5$ that contains a given vertex $v \in V(G)$ in $O(n+m)$ time, when such a cycle exists.  In fact, there is a $C_5$ that contains $v$ if and only if there are two adjacent vertices $w$ and $z$ at distance $2$ from $v$.  Indeed, both $N(vw)$ and $N(vz)$ are empty, because $G$ is triangle-free.  Then $v, a, b, w, z$ is a $C_5$ for any $a \in N(vw)$ and $b \in N(vz)$.  We sum up this procedure in Algorithm~\ref{alg:C5 slow}.  Its not hard to see that this algorithm takes $O(nm)$ time and $O(n+m)$ space.

\begin{algorithm}
  \caption{Induced $C_5$ in a triangle-free graph $G$.}\label{alg:C5 slow}
  \Input{a triangle-free graph $G$.}
  \Output{if existing, an induced $C_5$ of $G$; otherwise, a message.}

  \begin{AlgorithmSteps}
    \Step{For every $v \in V(G)$:}
    \begin{AlgorithmBlock}
      \Step{If there are two adjacent vertices $w, z$ at distance $2$ from $v$, then output $v, a, b, w, z$, for any $a \in N(vw)$ and $b \in N(vz)$, and halt.}
      \Step{Output ``$G$ contains no induced $C_5$'s.}
    \end{AlgorithmBlock}
  \end{AlgorithmSteps}
\end{algorithm}

\subsection{An $O(nm)$ time implementation of Step~\ref{alg:HBH recognition:C6}}

For the last step, suppose that $G$ is a degree ordered graph that is $C_4$-dominated and contains no triangles. The next lemma shows how to find an induced $C_6$ in $G$ that contains any given vertex $v \in V(G)$, if existing.

\begin{lemma}\label{lem:C6 slow}
  Let $G$ be a degree ordered graph that is $C_4$-dominated and contains no triangles, and $v_0 \in V(G)$.  Then, there is an induced $C_6$ in $G$ that contains $v_0$ if and only if there is a cycle $v_0, \ldots, v_5$ with the following properties:
  \begin{enumerate}[$(i)$]
    \item $v_3 \not\in N(v_0)$, 
    \item $v_1 = \min N(v_0v_2)$ and $v_5 = \min N(v_0v_4)$, 
    \item $v_0$ is dominated by neither $v_2$ nor $v_4$, and
    \item $v_1$ and $v_5$ are not dom-comparable.
  \end{enumerate}
\end{lemma}

\begin{proof}
  Suppose that $G$ contains a $C_6$ induced by $v_0, w_1, v_2, v_3, v_4, w_5$, in this order, and call $v_1 = \min N(v_0v_2)$ and $v_5 = \min N(v_0v_4)$.  If $w_1 \neq v_1$, then $v_0, v_1, v_2, w_1$ is a $C_4$, that is dominated by hypothesis.  Thus, $w_1$ dominates $v_1$ because $w_1 > v_1$ and $v_0$ is not dom-comparable with $v_2$.  Then, $v_0, \ldots, v_4, w_5$ is also an induced cycle.  Similarly, if $v_1 \neq w_1$, then $v_0,\ldots, v_5$ is an induced cycle as well.  Thus, $(i)$--$(iv)$ follow.

  For the converse, suppose that $v_0, \ldots, v_5$ is a cycle satisfying $(i)$--$(iv)$.  By $(i)$ and the fact that $G$ is triangle-free, the only possible edges between vertices of the cycle, besides those that are included in the cycle, are $v_1v_4$ and $v_2v_5$.  If $v_1$ is adjacent to $v_4$, then $v_0, v_4, v_5, v_1$ is a $C_4$ that, by $(iii)$ and $(iv)$, is not dominated.  Analogously, if $v_2$ is adjacent to $v_5$, then $v_0,v_1,v_2,v_5$ is a non-dominated $C_4$.  Finally, if neither $v_1v4$ nor $v_2v_5$ are edges of $G$, then $v_0, \ldots, v_5$ induce a $C_6$.
\end{proof}

The above lemma yields Algorithm~\ref{alg:C6 slow}, that finds an induced $C_6$, if existing.  With respect to the time complexity, as in Algorithm~\ref{alg:C4 slow}, Step \ref{alg:C6 slow:matrix} take $O(nm)$ time. A single traversal of $N(w_1)$, for every $w_1 \in N(w_0)$, is enough to compute Step~\ref{alg:C6 slow:N_2}.  As each access to $D$ takes $O(1)$ time, Step~\ref{alg:C6 slow:N_2} takes $O(n+m)$ time.  For Step~\ref{alg:C6 slow:tree}, first mark each vertex in $N(w_0)$.  Then, a single traversal of $N(w_2)$, for each $w_2 \in N_2$, while accessing the mark of the vertices, is enough to compute $p$ at Step~\ref{alg:C6 slow:tree}.   Thus, Step~\ref{alg:C6 slow:tree} also requires $O(n+m)$ time. For Step~\ref{alg:C6 slow:cycle}, first compute, for every $w_3 \not \in N(w_0)$, the ordering $u_1, \ldots, u_k$ of $N(w_3) \cap N_2$ so that $p(u_i) < p(u_{i+1})$ ($1 \leq i < k$).  If $w_2$, $w_4$ satisfy the conditions of Step~\ref{alg:C6 slow:cycle} then $u_i$ and $u_{i+1}$ also satisfy the conditions at it follows that Step~\ref{alg:C6 slow:cycle}, for some $1 \leq i < k$.  Then, as each access to $D$ and $p$ takes $O(1)$ time, Step~\ref{alg:C6 slow:cycle} takes $O(n+m)$ time. Therefore, the inner loop is executed in $O(nm)$ time, so the the time complexity of Algorithm~\ref{alg:C6 slow} is $O(nm)$.  For the spatial complexity, observe that matrix $D$ requires $O(n^2)$ bits, while all the other variables require at most $O(n+m)$ bits.  

\begin{algorithm}
  \caption{Induced $C_6$ in a $C_4$-dominated graph $G$ with no triangles.}\label{alg:C6 slow}
  \Input{a degree ordered graph $G$ that is $C_4$-dominated and contains no triangles.}
  \Output{if existing, an induced $C_6$ of $G$; otherwise, a message.}

  \begin{AlgorithmSteps}

    \Step{Let $v_1 > \ldots > v_n$ be the vertices of $G$.}\label{alg:C6 slow:ordering}
    \Step{Compute the matrix $D \in \{0,1\}^{n \times n}$ such that $d_{i,j} = 1$ if and only if $v_i$ dominates $v_j$, for $1 \leq i < j \leq n$.  Write $D(v_i, v_j) = d_{i,j}$.}\label{alg:C6 slow:matrix}
    \Step{For each $w_0 \in V(G)$, do:}
    \begin{AlgorithmBlock}
      \Step{Compute $N_2 := \{w_2 \in V(G) \mid N(w_0w_2) \neq \emptyset \text{ and } D(w_2,w_0) = 0\}$}\label{alg:C6 slow:N_2}
      \Step{For each $w_2 \in N_2$, set $p(w_2) := \min N(w_0w_2)$.}\label{alg:C6 slow:tree}
      \Step{If there is a vertex $w_3 \not\in N(w_0)$ that is adjacent to $w_2, w_4 \in N_2$ and $D(p(w_2), p(w_4)) + D(p(w_4), p(w_2)) = 0$, then output $w_0, p(w_2), w_2, w_3, w_4, p(w_4)$ and halt.}\label{alg:C6 slow:cycle}
    \end{AlgorithmBlock}
    \Step{Output ``$G$ contains no induced $C_6$'s''.}
  \end{AlgorithmSteps}
\end{algorithm}

\section{Faster recognition of $C_4$-dominated graphs with no triangles}
\label{sec:C4 fast}

In this section we develop an improved implementation of Algorithm~\ref{alg:C4 slow} whose running time and space consumption are $O(m\alpha(G))$ and $O(m)$, respectively.  The idea is the same, for each $S \in \MS$ we first check if $S$ is safe.  If not, then we check the dominations between the vertices in $L(S)$.  The major difference is that the domination matrix $D$ is no longer employed; instead, while checking the safeness, we compute a digraph that encodes some dominations of interest of $G$.  We begin with the description of this digraph.  As before, we assume that $G$ is a degree ordered graph with no triangles.

Fix $S \in \MS$.  Define the binary relation $\to_S$ on $L(S)$ as follows: for $a,b \in L(S)$, $a \to_S b$ if and only if $a < b$ and there is no $c \in L(S)$ such that $a < c < b$.  In other words, $\to_S$ defines the subordering of $<$ induced by the members in $L(S)$.  This subordering is exactly the same that is given by method $C4$ (see Section~\ref{sec:slow}).  We require a generalization of $\MS$ from vertices to sets.  For $V \subseteq V(G)$, define $\MS(V) = \bigcup_{v \in V}{\MS(v)}$.  Note that $\MS(\emptyset) = \emptyset$ and $\MS(V(G)) = \MS(G)$.  Now we are ready to define the domination digraph.

For $V \subseteq V(G)$, the \emph{unsafe domination digraph} of $V$ is the digraph $U(V)$ with vertex set $V(G)$ such that, for $a,b \in V(G)$, $a \to_D b$ if and only if $a \to_S b$ for some unsafe $S \in \MS(V)$.  In other words, for each unsafe $S \in \MS(V)$, there is a path that goes through the vertices of $L(S)$ in the order given by $<$.  The \emph{unsafe domination digraph} of $G$ is the digraph $U(G) = U(V(G))$.  The following lemma shows that, when every $S \in \MS(V)$ is dominated, $U(V)$ is actually a directed forest that encodes some dominations.

\begin{lemma}\label{lem:domination tree}
  Let $G$ be a degree ordered graph that is $C_4$-dominated and contains no triangles, and $V \subseteq V(G)$.  If $S$ dominated for every $S \in \MS(V)$, then the following conditions hold for any $v \in V(G)$:
  \begin{enumerate}[$(i)$]
    \item $d^+_{U(V)}(v) \leq 1$, and
    \item $v$ is dominated in $G$ by all its out-neighbors.
  \end{enumerate}
\end{lemma}

\begin{proof}
  $(i)$. Suppose that $v \in V(G)$ has at least two out-neighbors $a > b$ in $U(V)$.  Then, by definition, there are two unsafe triples in $\MS(V)$, say $S$ and $T$, such that $v \in L(S) \cap L(T)$, $v \to_S a$, and $v \to_T b$.  Now, since $a > b$ and $v \to_S a$, it follows that $b < v(S)$ and $b \not \in L(S)$.  Consequently, $b$ is not adjacent to either $v(S)$ or $w(S)$, thus $v$ is not dominated by $b$ in $G$.  Also, since $v \to_T b$, it follows that $v < b$, i.e. $d_G(b) \geq d_G(v)$, hence $b$ is neither dominated by $v$ in $G$.  But then, $T \in \MS(V)$ is not dominated.

  $(ii)$. If $d^+_{U(V)}(v) = 0$, then $(ii)$ is vacuously true.  Otherwise, let $w$ be an out-neighbor of $v$.   Since $v \to_{U(V)} w$, we obtain that $v \to_S w$, for some unsafe $S \in \MS(V)$.  Since $S$ is dominated and unsafe, it follows that $w$ dominates $v$ in $G$.
\end{proof}

Recall that the idea of the new implementation is to build $U(G)$ so as to test the dominations inside $L(S)$, for every unsafe $S \in \MS$.  In turn, to build $U(G)$, we need to know the safeness status of some triples in $\MS$, for which we require the dominations of these triples.  It turns out that $U(G)$ can be iteratively computed while the partial results are used to check the safeness of those triples of interest.  The next lemma shows how this construction is done.  

\begin{lemma}\label{lem:C4 algorithm}
  Let $G$ be a degree ordered graph that is $C_4$-dominated and contains no triangles, and $S \in \MS$.  Call $v = v(S)$, $w = w(S)$, and $L = L(S)$.  If every triple of $\MS(\MAX(v))$ is dominated, then the following are equivalent statements. 
  \begin{enumerate}[$(i)$]
    \item $S$ is safe. 
    \item $|L| = |N(w) \cap \MIN(v)|$ and if $w \in L(T)$ for some unsafe $T \in \MS(\MAX(v))$, then there is a path from $w$ to $v$ in $U(\MAX(v))$.
  \end{enumerate}
\end{lemma}

\begin{proof}
  $(i) \Longrightarrow (ii)$.  Suppose first that $S$ is safe, i.e., $v$ dominates $w$.  By definition, $L = N(vw) \cap \MIN(v)$, thus $|L| = |N(w) \cap \MIN(v)|$.  Now, suppose that there is some unsafe $T \in \MS(\MAX(v))$ such that $w \in L(T)$.  In this case, since $v$ dominates $w$, it follows that $v$ is adjacent to both $v(T)$ and $w(T)$.  Hence $v \in L(T)$, because $v < v(T)$ as $T \in \MS(\MAX(v))$.  Therefore, there is a path from $w$ to $v$ in $U(\MAX(v))$ because $v > w$.

  $(ii) \Longrightarrow (i)$. In this case, we argue by contradiction.  Suppose that $(ii)$ is true and yet $S$ is unsafe, i.e., $N(w) \setminus N(v)$ contains some vertex $z$.  Since $L = N(vw) \cap \MIN(v) \subseteq N(v)$ and $|L| = |N(w) \cap \MIN(v)|$, we obtain that $L = N(w) \cap \MIN(v)$, thus $z > v$.  Fix $a \in L$.  Since $w$ is adjacent to both $z$ and $a$, it follows that $w \in L(z,a)$, thus $(z, a) \in \MS(z)$.  Also, since $v \in N(a) \setminus N(z)$, it follows that $(z, a)$ is unsafe, so there is a path from $w$ to $v$ in $U(\MAX(v))$.  Therefore, by Lemma~\ref{lem:domination tree}, $v$ dominates $w$, a contradiction.
\end{proof}

Lemma~\ref{lem:C4 algorithm} yields Algorithm~\ref{alg:C4 fast}.  Its input is the graph $G$, and its output is $U(G)$ if $G$ is $C_4$-dominated, or a non dominating cycle otherwise.  Thus, Algorithm~\ref{alg:C4 fast} is just a replacement of Algorithm~\ref{alg:C4 slow}. We discuss its correctness and complexity in the next paragraphs.

\begin{algorithm}[htb!]
  \caption{Non-dominated $C_4$ of a triangle-free graph $G$.}\label{alg:C4 fast}
  \Input{a degree ordered graph $G$ with no triangles.}
  \Output{if $G$ is $C_4$-dominated, then $U(G)$; otherwise, a non-dominated $C_4$ of $G$.}

  \begin{AlgorithmSteps}
    \Step{Let $v_1 > \ldots > v_n$ be the vertices of $G$.}\label{alg:C4 fast:ordering}
    \Step{Set $UNSAFE := \emptyset$, $U$ $:=$ $(V(G), \emptyset)$ and $d^<(v) := |N(v)|$, for every $v \in V(G)$.}\label{alg:C4 fast:initialization}
    \Step{For $i := 1, \ldots, n$, do:}\label{alg:C4 fast:out loop begin}
    \begin{AlgorithmBlock}
      \Step{Set $REACH := \{w \in V(G) \mid d^<(w) >0 \text{ and there is a path from $w$ to $v_i$ in $U$}\}$.}\label{alg:C4 fast:path}
      \Step{For each $S \in \{(v_i,w,L) \in \MS \mid \text{either } |L| \neq d^<(w) \text{ or } (w \in UNSAFE \text{ and } w \not\in REACH)\}$.}\label{alg:C4 fast:domination begin}
      \begin{AlgorithmBlock}
        \Comment{This loop checks that every unsafe $S$ is dominated.}
        \Step{For each $a,b \in L(S)$ such that $a \to_S b$ do:}
        \begin{AlgorithmBlock}
          \Step{If $N^+_U(a) \not\subset \{b\}$, then output $v_i, a, w(S), b$ and halt.}\label{alg:C4 fast:halt 1}
          \Step{If $N^+_U(a) = \emptyset$ and $a$ is not dominated by $b$, then output $v_i, a, w(S), b$ and halt.}\label{alg:C4 fast:halt 2}
          \Step{Add $a \to b$ to $U$.}\label{alg:C4 fast:parent setting}
        \end{AlgorithmBlock}
        \Step{Set $UNSAFE := UNSAFE \cup L(S)$.}\label{alg:C4 fast:unsafe setting}\label{alg:C4 fast:domination end}
      \end{AlgorithmBlock}
      \Step{Set $d^<(w) := d^<(w)-1$, for every $w \in N(v_i)$.}\label{alg:C4 fast:degree setting}\label{alg:C4 fast:out loop end}
    \end{AlgorithmBlock}
    \Step{Output $U$.}\label{alg:C4 fast:unsafe digraph}
  \end{AlgorithmSteps}
\end{algorithm}

\subsection{Correctness of Algorithm~\ref{alg:C4 fast}}

The Loop \ref{alg:C4 fast:out loop begin}--\ref{alg:C4 fast:out loop end} examines the vertices in the order defined by $<$, beginning from the greater.  Step~\ref{alg:C4 fast:initialization} initializes some variables.  Fix $i \in \{1, \ldots, n\}$, and call $MAX$ $=$ $\MAX(v_i)$, and $MIN = \MIN(v_i)$.  Observe that, by Step~\ref{alg:C4 fast:ordering}, $MAX$ $=$ $\{v_1, \ldots, v_{i-1}\}$, and $MIN$ $=$ $\{v_{i+1}, \ldots, v_{n}\}$.  Immediately before Loop \ref{alg:C4 fast:out loop begin}--\ref{alg:C4 fast:out loop end} is executed for $v_i$, every $S \in \MS(MAX)$ is dominated, and the state of the variables is as follows:
\begin{enumerate}[$(i)$]
  \item $UNSAFE = \{v \in V(G) \mid v \in L(S) \text{ for some unsafe } S \in \MS(MAX)\}$.
  \item $U = U(MAX)$.  
  \item $d^<(v) = |N(v) \cap (MIN \cup \{v_i\})|$, for every $v \in V(G)$.
\end{enumerate}

Step~\ref{alg:C4 fast:path} finds all those $w \in V(G)$ such that there is a path from $w$ to $v_i$ in $U(MAX)$, and $d^<(w) > 0$.  Observe that if $d^<(w) = 0$, then, by $(iii)$, $L(v_i, w) = \emptyset$, thus $(v_i, w) \not\in \MS$.  Then, by  Lemma~\ref{lem:C4 algorithm}, Loop \ref{alg:C4 fast:domination begin}--\ref{alg:C4 fast:domination end} iterates every unsafe triple $S \in \MS$ whose high vertex is $v_i$.  This loop is the responsible for testing whether $S$ is dominated or not, and it has the following three alternatives.
\begin{description}
  \item[Alternative 1:] the algorithm halts at Step~\ref{alg:C4 fast:halt 1} while examining $a \to_S b$.  For this to happen, $N^+_U(a)$ must contain a vertex different than $b$, say $\mu(a)$.  Thus, by $(ii)$, both $b$ and $\mu(a)$ are out-neighbors of $a$ in $U(MAX \cup \{v_i\})$.  Therefore, $G$ is not $C_4$-dominated by Lemma~\ref{lem:domination tree}.
  \item [Alternative 2:] the algorithm halts at Step~\ref{alg:C4 fast:halt 2} while examining $a \to_S b$.  This alternative occurs when $a$ is not dominated by $b$, thus $S$ is not dominated. Therefore, $G$ is not $C_4$-dominated.  
  \item [Alternative 3:] the algorithm does not halt inside Loop \ref{alg:C4 fast:domination begin}--\ref{alg:C4 fast:domination end}.  In this last alternative, for every $a \to_S b$, either $N^+_U(a) = \{b\}$ or $a$ is dominated by $b$.  Whichever the case, by $(ii)$, $a$ is dominated by $b$ and $a \to b$ is an edge of $U(MAX \cup \{v\})$.  In particular, $S$ is dominated.
\end{description}
Therefore, if Loop \ref{alg:C4 fast:domination begin}--\ref{alg:C4 fast:domination end} halts if and only if $\MS(MAX \cup \{v_i\})$ contains a non-dominated triple.  Furthermore, if Loop \ref{alg:C4 fast:domination begin}--\ref{alg:C4 fast:domination end} does not halt, then by Step~\ref{alg:C4 fast:parent setting}, $(ii)$ is satisfied before the execution of the outer loop for $v_{i+1}$.  Finally, by Steps \ref{alg:C4 fast:unsafe setting}~and~\ref{alg:C4 fast:degree setting}, $(i)$ and $(iii)$ also hold immediately before the execution of Loop \ref{alg:C4 fast:out loop begin}--\ref{alg:C4 fast:out loop end} for $v_{i+1}$.  (Observe that Step~\ref{alg:C4 fast:halt 1} is superfluous, and it can be removed from the algorithm without affecting its correctness. However, its inclusion drops the time complexity required by the algorithm.  In some sense, it tells us that the domination of $a$ by $b$ was already tested.)

Algorithm~\ref{alg:C4 fast} gives its output in one of three steps.  Suppose that the algorithm halts at Step~\ref{alg:C4 fast:halt 1}, as in Alternative 1.  This happens because there is an unsafe triple $T$, already processed by the algorithm, such that $a \to_T \mu(a)$.  We claim that $b < \mu(a)$; otherwise, as in Lemma~\ref{lem:domination tree}, $T$ would not be dominated, contradicting the fact that the Algorithm~\ref{alg:C4 fast} stops immediately after it process a non-dominated triple.  Then, as in Lemma~\ref{lem:domination tree}, $a$ is not dominated by $b$.  Similarly, if Algorithm~\ref{alg:C4 fast} halts at Step~\ref{alg:C4 fast:halt 2}, as in Alternative 2, then $a$ is not dominated by $b$.  Therefore, in both of these alternatives, $v_i, a, w(S), b$ is a non-dominated $C_4$.  Finally, if Algorithm~\ref{alg:C4 fast} does not halt inside Loop \ref{alg:C4 fast:domination begin}--\ref{alg:C4 fast:domination end}, as in Alternative 3, then $U = U(G)$, by $(ii)$. Summing up, Algorithm~\ref{alg:C4 fast} is correct.

\subsection{Implementation and complexity of Algorithm~\ref{alg:C4 fast}}

The implementation of Algorithm~\ref{alg:C4 fast} is rather straightforward.  Recall that, by Lemma~\ref{lem:domination tree}, every vertex of $U$ has at most one out-neighbor.  We record such out-neighbors in a vector with $n$ positions, where the $i$-th position is $b$ if and only if $v_i \to_U b$.  If $N^+(v_i) = \emptyset$, then the $i$-th position of the array is some undefined value $\bot$.  Similarly, $d^<$ and $UNSAFE$ are also stored in vectors with $n$ positions; the $i$-th position respectively indicates the values of $d^<(v_i)$ and $v_i \in UNSAFE$.  We assume that $a$ can be inserted into $N^-_U(b)$ in $O(1)$ time at Step~\ref{alg:C4 fast:parent setting}, and that $a$ can be also removed from $N^-_U(b)$ in $O(1)$ time when $d^<(a)$ drops from $1$ to $0$ at Step~\ref{alg:C4 fast:degree setting}.  This can be achieved by storing $N^-_U(b)$, for every $b \in V(G)$, and a pointer from $a$ to its position in $N^-_U(b)$, for every $a \in N^-_U(b)$. 

Discuss the time complexity of the algorithm with the above implementation.  Fix $i \in \{1, \ldots, n\}$.  Before entering the Loop \ref{alg:C4 fast:out loop begin}--\ref{alg:C4 fast:out loop end} for $v_i$, $\MS(v_i)$ is obtained by executing one iteration of method $C4$.  As proved in~\cite{ChibaNishizekiSJC1985}, this iteration of $C4$ costs $O(\sum_{S \in \MS(v_i)}{|L(S)|})$ time.  For Step~\ref{alg:C4 fast:path}, just use a tree traversal algorithm, taking advantage that $N^-_U(b)$ is stored for every $b \in V(G)$.  Such a traversal takes $O(1)$ time per vertex traversed.  Observe that $w$ is traversed at Step~\ref{alg:C4 fast:path} if and only if $L(v_i, w) \neq \emptyset$, thus Step~\ref{alg:C4 fast:path} takes $O(|\MS(v_i)|)$ time.  Finally, Loop  \ref{alg:C4 fast:domination begin}--\ref{alg:C4 fast:domination end} takes $O(|L(S)|)$ time to examine each $S \in \MS(v_i)$.  Thus, this loop also requires $O(\sum_{S \in \MS(v_i)}{|L(S)|})$.  Finally, all the steps outside the Loop \ref{alg:C4 fast:out loop begin}--\ref{alg:C4 fast:out loop end} cost $O(n)$ time.  Therefore, as proven in~\cite{ChibaNishizekiSJC1985}, the time complexity of Algorithm~\ref{alg:C4 fast} is \[O\left(n+\sum_{v \in V(G)}\sum_{S \in \MS(v)}{|L(S)|}\right) = O(n+m\alpha(G)).\]

On the other hand, only \[O\left(n+\max_{v \in V(G)} \left\{\sum_{S \in \MS(v)}{|L(S)|}\right\}\right) = O(n+m)\] bits of additional space are used by the algorithm, thus the space complexity is $O(n+m)$.

\section{Faster recognition of hereditary biclique-Helly graphs}
\label{sec:C6 fast}

In this section we improve Algorithms \ref{alg:C5 slow}~and~\ref{alg:C6 slow} so as to run in $O(m\alpha(G) + n^2)$ time and $O(m)$ space.  Again, the idea is to evaluate if a given vertex $v$ belongs to an induced $C_5$ or an induced $C_6$, by looking at those vertices at distance at most $3$ from $v$.  In this section, however, we take advantage of the dominations implied by the squares family of $G$.  The implementation of the improvement of Algorithm~\ref{alg:C5 slow} for finding an induced $C_5$ is rather similar to the implementation of the improvement of Algorithm~\ref{alg:C6 slow} for finding an induced $C_6$.  So, we only describe in detail the improvement of Algorithm~\ref{alg:C6 slow}, and briefly discuss the improvement of Algorithm~\ref{alg:C5 slow}. The main tool in this section is a new forest that extends the unsafe domination digraph.  For the rest of this section, suppose that $G$ is a degree ordered graph that is $C_4$-dominated and has no triangles.

The \emph{squares domination digraph} of $G$ is the digraph $S(G)$ that is obtained from $U(G)$ by inserting an edge $w \to v$ for every safe $(v,w) \in \MS$.  Fix $v \in V(G)$.  Define $\sigma(v) = \min N^+_{S(G)}(v)$, when $d^+_{S(G)}(v) > 0$.  We write $\sigma(v) = \bot$ to indicate that $d^+_{S(G)}(v) = 0$.  For the sake of simplicity, we assume $\bot > \max V(G)$.  

The following lemma is analogous to Lemma~\ref{lem:C6 slow}.  It shows how to find an induced $C_6$, when one such cycle exists.  Observe the role that $\sigma$ plays by marking that $v_2$ and $v_4$ are not dom-comparable (condition $(ii)$), and that $v_1$, $v_3$, and $v_5$ are neither dom-comparable (condition $(iii)$).  Indeed, $\sigma(v_i)$ was chosen as the minimum that dominates $v_i$ in either $U(G)$ or in a safe triple.  So, if $v < \sigma(w)$, and $v$ and $w$ share a $C_4$, then $v$ cannot dominate $w$.

\begin{lemma}\label{lem:C6 fast}
  Let $G$ be a degree ordered graph that is $C_4$-dominated and contains no triangles.  Then, $G$ contains an induced $C_6$ if and only if it contains two paths $v_0, v_1, v_2, v_3$ and $v_0, v_5, v_4, v_3$ such that
  \begin{enumerate}[$(i)$]
    \item $v_2 \neq v_4$, $v_3 \not \in N(v_0)$ and $v_0 > \max\{v_1, v_2, v_4, v_5\}$,
    \item $\sigma(v_2) > v_0$ and $\sigma(v_4) > v_0$, and
    \item $\sigma(v_1) = \sigma(v_3) = \sigma(v_5)$.
  \end{enumerate}
  Furthermore, if $G$ contains the paths $v_0, v_1, v_2, v_3$ and $v_0, v_5, v_4, v_3$ satisfying $(i)$--$(iii)$, then $v_0, \ldots, v_5$ induce a cycle in $G$.
\end{lemma}

\begin{proof}
  Suppose first that $G$ contains an induced $C_6$.  For each vertex $v$, define $r(v)$ as the number of vertices reached from $v$ in $S(G)$.  For each $C \subseteq V(G)$, define $R(C) = \sum_{v \in C}{r(v)}$.

  From among all the induced $C_6$'s of $G$, chose the cycle $v_0, \ldots, v_5$ such that:
  \begin{enumerate}[(1)]
    \item there is no induced $C_6$ of $G$ containing a vertex in $\MAX(v_0)$, and
    \item $R(\{v_0, \ldots, v_5\})$ is minimum among those cycles containing $v_0$.
  \end{enumerate}
  We claim that paths $v_0, v_1, v_2, v_3$ and $v_0, v_5, v_4, v_3$ satisfy $(i)$--$(iii)$.  Clearly, $v_2 \neq v_4$, $v_3 \not\in N(v_0)$, and $v_0 > \max\{v_1, \ldots, v_5\}$, thus $(i)$ follows.

  Suppose, contrary to statement $(ii)$, that $\sigma(v_2) \leq v_0$.  Notice that $v_2$ is not dominated by $v_0$, so $\sigma(v_2) < v_0$.  Then, $\sigma(v_2) \neq \bot$, hence $\sigma(v_2)$ dominates $v_2$.  By (2), $v_0, v_1, \sigma(v_2), v_3, v_4, v_5$ do not induce a cycle in $G$, because $r(v_2) > r(\sigma(v_2))$. Thus, as $G$ is triangle-free, we obtain that $\sigma(v_2)$ is adjacent to $v_1$, $v_3$, and $v_5$, hence, $v_0, v_1, \sigma(v_2), v_5$ is a $C_4$.  Now, since $G$ is $C_4$-dominated and $v_1$ and $v_5$ are not dom-comparable, it follows that $\sigma(v_2)$ dominates $v_0$.  However, this contradicts the fact that $\sigma(v_2) < v_0$.  Similar argument prove that $\sigma(v_4) > v_0$, so $(ii)$ follows.

  Finally, consider statement $(iii)$.  Clearly, $(iii)$ is true when $\sigma(v_1) = \sigma(v_3) = \sigma(v_5) = \bot$.  Suppose, then, that $\sigma(v_i) \neq \bot$, for some $i \in \{1,3,5\}$.  In this case, $\sigma(v_i)$ dominates $v_i$, by definition of $\sigma$.  Again, by (2), we obtain that $\sigma(v_i)$ must be adjacent to $v_0, v_2$, and $v_4$, or otherwise $v_i$ could be replaced by $\sigma(v_i)$ so as to obtain an induced $C_6$ with lower value of $R$.  Then, $\{v_0, v_1, v_2, \sigma(v_i)\}$, $\{v_0, v_5, v_4, \sigma(v_i)\}$, and $\{v_2, v_3, v_4, \sigma(v_i)\}$ all induce $C_4$'s.  Since $G$ is $C_4$-dominated and $v_0, v_2, v_4$ are pairwise not dom-comparable, it follows that $\sigma(v_i)$ dominates $v_1$, $v_3$, and $v_5$.  Thus, none of $\sigma(v_1)$, $\sigma(v_3)$, and $\sigma(v_5)$ is equal to $\bot$.  Repeating the above arguments for $i$ $=$ $1$, $3$, and $5$, we obtain that $\sigma(v_i)$ dominates $v_j$ for every $i,j \in \{1,3,5\}$.  Then $(iii)$ follows, because $\sigma(v) = \min N^+_{S(G)}(v)$, for every $v \in V(G)$.

  For the converse, suppose that there are two paths $v_0, v_1, v_2, v_3$, and $v_0, v_5, v_4, v_3$ that satisfy $(i)$--$(iii)$.   Observe that it is enough to prove that $v_1 \not\in N(v_4)$ and $v_5 \not\in N(v_2)$.  Indeed, in this case $v_1 \neq v_5$, thus $v_0, \ldots, v_5$ is a cycle of $G$ and, as $G$ is triangle-free, this cycle is induced.

  Suppose, to obtain a contradiction, that $v_1 \in N(v_4)$.  Then, $v_1$, $v_2$, $v_3$, $v_4$ is a $C_4$ in $G$.  This cycle is represented by some $S \in \MS$.  Consider the following cases.
  \begin{description}
    \item[Case 1:] $v(S) \in \{v_2, v_4\}$.  Observe that $w(S) \in \{v_2,v_4\}$, so $\sigma(w(S)) > v_0 > v(S)$ by $(i)$ and $(ii)$.  Then, $v(s)$ does not dominate $w(S)$, thus $S$ is an unsafe triple of $\MS$.  Consequently, since $v_1, v_3 \in L(S)$ and $v_0 \in N(v_1) \setminus N(v_3)$, it follows that $v_1$ dominates $v_3$, thus there is a path from $v_3$ to $v_1$ in $U(G)$.  Then, $\sigma(v_1) > v_1 \geq \sigma(v_3)$, a contradiction.
    \item[Case 2:] $v(S) \not\in \{v_2,v_4\}$.  In this case, both $v_2, v_4$ belong to $L(S)$, thus, as $\sigma(v_2) > v_0 > v_4$ and $\sigma(v_4) > v_0 > v_2$, we obtain that $S$ must be safe.  Consequently, since $v_0 \in N(v_1) \setminus N(v_3)$, it follows that $S = (v_1, v_3)$.  This is again impossible, because $\sigma(v_1) > v_1 \geq \sigma(v_3)$.
  \end{description}
  Replacing $v_1$ by $v_5$ above, we obtain that $v_5 \not\in N(v_2)$ as well.  Thus, $v_0, \ldots, v_5$ is an induced cycle of $G$.
\end{proof}

Lemma~\ref{lem:C6 fast} implies Algorithm~\ref{alg:C6 fast}.  Its input is the graph $G$, and its output is a message if $G$ contains no induced $C_6$, or an induced $C_6$ otherwise.  Thus, Algorithm~\ref{alg:C6 fast} is just a replacement of Algorithm~\ref{alg:C6 slow}. We discuss the algorithm, its correctness, and its complexity in the next paragraphs.

\begin{algorithm}
  \caption{Induced $C_6$ of a triangle-free $C_4$-dominated graph.}\label{alg:C6 fast}
  \Input{a degree ordered graph $G$ that is $C_4$-dominated and has no triangles.}
  \Output{if existing, an induced $C_6$ of $G$; otherwise, a message.}

  \begin{AlgorithmSteps}
    \Step{Let $v_1 > \ldots > v_n$ be the vertices of $G$.}
    \Step{Compute $\sigma(v)$, for every $v \in V(G)$.}\label{alg:C6 fast:S(G)}
    \Step{For each $vw \in E(G)$, set $X(v,w) := \{z \in N(v) \mid \sigma(z) = w\}$ and $X(w,v) := \{z \in N(w) \mid \sigma(z) = v\}$.}\label{alg:C6 fast:X}
    \Step{For $i := 1$ to $n$, do:}\label{alg:C6 fast:OutLoopBegin}
    \begin{AlgorithmBlock}
      \Step{Set $N_1 := \{w_1 \in N(v_i) \mid w_1 < v_i\}$.}\label{alg:C6 fast:N1}
      \Step{Set $N_2 := \bigcup_{w_1 \in N_1}\{w_2 \in N(w_1) \mid \sigma(w_2) > v_i > w_2\}$.}\label{alg:C6 fast:N2}
      \Step{For each $w_2 \in N_2$, set $N_3(w_2) := \bigcup\{X(w_2, \sigma(w_1)) \setminus N_1 \mid w_1 \in N_1 \cap N(w_2)\}$.}\label{alg:C6 fast:N3}
      \Step{For $w_2, w_4 \in N_2$ ($w_2 \neq w_4$), if $N_3(w_2) \cap N_3(w_4)$ contains a vertex, say $w_3$, then:}\label{alg:C6 fast:condition}
      \begin{AlgorithmBlock}
        \Step{Output $v_i, w_1, w_2, w_3, w_4, w_5$, for some $w_1, w_5 \in N_1$ such that $w_1 \neq w_5$ and $\sigma(w_1) = \sigma(w_5) = \sigma(w_3)$, and halt.}\label{alg:C6 fast:halt}\label{alg:C6 fast:OutLoopEnd}
      \end{AlgorithmBlock}
    \end{AlgorithmBlock}
    \Step{Output ``$G$ contains no induced $C_6$'s.''}
  \end{AlgorithmSteps}
\end{algorithm}

\subsection{Correctness of Algorithm~\ref{alg:C6 fast}}

Algorithm~\ref{alg:C6 fast} either halts at Step~\ref{alg:C6 fast:halt} or it runs until its termination.  In this paragraph, we prove that Algorithm~\ref{alg:C6 fast} halts at Step~\ref{alg:C6 fast:halt} if and only if $G$ contains an induced $C_6$.  Furthermore, when Algorithm~\ref{alg:C6 fast} halts at Step~\ref{alg:C6 fast:halt}, its output is an induced $C_6$ of $G$.

Suppose first that Algorithm~\ref{alg:C6 fast} halts at Step~\ref{alg:C6 fast:halt}, when the $i$-th iteration of Loop \ref{alg:C6 fast:OutLoopBegin}--\ref{alg:C6 fast:OutLoopEnd} is being executed. Examine the state of the variables immediately before Step~\ref{alg:C6 fast:halt} is executed.  By Steps~\ref{alg:C6 fast:X}, \ref{alg:C6 fast:N1}~and~\ref{alg:C6 fast:N3}, $w_3 \in N(w_2w_4) \setminus N(v_i)$, and there are two vertices $w_1 \in N(v_iw_2)$ and $w_5 \in N(v_iw_4)$ such that $\sigma(w_1) = \sigma(w_3) = \sigma(w_5)$.  By Steps \ref{alg:C6 fast:N1}~and~\ref{alg:C6 fast:N2}, $v_i > \max\{w_1, w_2, w_4, w_5\}$, so, by Lemma~\ref{lem:C6 fast}, $v_i, w_1, \ldots, w_5$ induce a $C_6$ in $G$.  

For the converse, suppose that $G$ contains an induced $C_6$.  By Lemma~\ref{lem:C6 fast}, there must be two paths $w_0, w_1, w_2, w_3$ and $w_0, w_5, w_4, w_3$ satisfying conditions $(i)$--$(iii)$ of the lemma.  Vertex $w_0$ gets some name $v_i$ in Algorithm~\ref{alg:C6 fast}, for some $1 \leq i \leq n$.  If Algorithm~\ref{lem:C6 fast} halts before the $i$-th iteration of Loop \ref{alg:C6 fast:OutLoopBegin}--\ref{alg:C6 fast:OutLoopEnd}, then there is nothing to prove.  So, suppose that the $i$-th iteration of Loop \ref{alg:C6 fast:OutLoopBegin}--\ref{alg:C6 fast:OutLoopEnd} is executed, and consider its effects.  By Step~\ref{alg:C6 fast:N1} and Lemma~\ref{lem:C6 fast}~$(i)$, $w_1, w_5 \in N_1$.  By Step~\ref{alg:C6 fast:N2} and Lemma~\ref{lem:C6 fast} $(i)$--$(ii)$, $w_2, w_4 \in N_2$.  By Step~\ref{alg:C6 fast:X} and Lemma~\ref{lem:C6 fast}~$(iii)$, $w_3$ belongs to both $X(w_2, \sigma(w_1))$ and $X(w_4, \sigma(w_5))$, thus, by Step~\ref{alg:C6 fast:N3} and Lemma~\ref{lem:C6 fast}~$(i)$, $w_3 \in N_3(w_2) \cap N_3(w_4)$.  Therefore, the condition of Step~\ref{alg:C6 fast:condition} is satisfied, and Algorithm~\ref{alg:C6 fast} halts at Step~\ref{alg:C6 fast:halt}.

For the furthermore part, observe that Step~\ref{alg:C6 fast:halt} always finds the vertices $w_1 \neq w_5$ such that $\sigma(w_1) = \sigma(w_5) = \sigma(w_3)$.  Indeed, such $w_1$ and $w_5$ exist as argued before. And, by Lemma~\ref{lem:C6 fast}, $v_i, w_1, \ldots, w_5$ induce $C_6$ in $G$.  Summing up, Algorithm~\ref{alg:C6 fast} is correct.

\subsection{Implementation and complexity of Algorithm~\ref{alg:C6 fast}}

The data structures involved in Algorithm~\ref{alg:C6 fast} are a little harder than those in Algorithm~\ref{alg:C4 fast}.  The input graph $G$ is implemented with adjacency lists, i.e., the list $N(v)$ is stored for each $v \in V(G)$.  While $N(v)$ is being iterated, say the next vertex is $w$, $O(1)$ time access to the position that $v$ occupies in $N(w)$ is required.  This can be achieved by keeping a pointer to this position paired with $w$ in $N(v)$.  The value of $\sigma(v)$ is stored together with the vertex $v$.  Given a vertex $v$, $O(1)$ time access to the set $children(v) = \{w \in V(G) \mid \sigma(w) = v\}$ is also required.  So, a list with the elements of $children(v)$ is stored together with $v$.  A list with the elements of $X(v,w)$ is stored for every $vw \in E(G)$.  The list $X(v,w)$ has to be obtained in $O(1)$ time while $w$ is being examined in a traversal of $N(v)$.  Thus, $X(v,w)$ is also paired with $w$ in $N(v)$.  Furthermore, for each $z \in X(v,w)$, while traversing $N(z)$ with $v$ as the next vertex, the list $X(v,w)$ has to be accessed in $O(1)$ time.  So, a reference $x(z,v)$ pointing to $X(v,w)$ is paired together with $v$ in the list $N(z)$.  Recall that $\{X(v,w)\}_{w \in N(v)}$ is a partition of $N(v)$; indeed, $z$ belongs only to $X(v,\sigma(z))$.  Thus there is a unique pointer $x(z,v)$ associated with $v$ in $N(z)$. Finally, sets $N_1, N_2$ are implemented as $n$-position vectors so that membership can be tested in $O(1)$ time.

Consider the time complexity of Algorithm~\ref{alg:C6 fast}.  To compute $\sigma$ at Step~\ref{alg:C6 fast:S(G)} we first set $\sigma(v)$ as the out-neighbor of $v$ in $U(G)$ by calling Algorithm~\ref{alg:C4 fast}; next, we traverse the family $\MS$ given by method $C4$ in~\cite{ChibaNishizekiSJC1985} so as to update the values of $\sigma$ accordingly.  All these steps take $O(m\alpha(G))$ time.  Next, by traversing $V(G)$, we compute $children(v)$ in $O(n)$ time.  For computing the sets $X(v,w)$ at Step~\ref{alg:C6 fast:X}, for a given $v \in V(G)$, we use a three step procedure.  First, we traverse $N(v)$ and set $X(v,w) = \emptyset$, for each $w \in N(v)$.  Second, we sort $N(v)$ according to the value of $\sigma$, i.e., $w$ appears before $z$ in $N(v)$ only if $\sigma(w) \geq \sigma(z)$.  Third, we traverse $N(v)$ once again and, for each $w \in N(v)$, we insert $w$ into $X(v, \sigma(w))$.  These steps take $O(d(v))$ time each; in particular, the third step can be done as in a merging procedure, because $N(v)$ is sorted according to the values of $\sigma$.  Therefore, Step~\ref{alg:C6 fast:X} takes $O(n+m)$ time.  Before executing Step~\ref{alg:C6 fast:OutLoopBegin}, we should update the values of the pointers $x(z,v)$.  For this, we traverse each $N(v)$ once again as in the merging step and, when $z$ is being traversed so as to be inserted in $X(v,w)$, we access the position of $v$ inside $N(z)$ in $O(1)$ time, and set $x(z,v)$ to point to $X(v,w)$.  Therefore, the update of the $x$ pointers takes $O(n+m)$ time as well.  To evaluate the time required by Loop \ref{alg:C6 fast:OutLoopBegin}--\ref{alg:C6 fast:OutLoopEnd}, consider a vertex $v_i \in V(G)$.  Step~\ref{alg:C6 fast:N1} requires only a traversal of $N(v_i)$, so it takes $O(d(v_i))$. For Step~\ref{alg:C6 fast:N2}, it is enough to traverse $N(w)$, for each $w \in N_1$, while $\sigma$ is accessed in $O(1)$ time.  Therefore, Step~\ref{alg:C6 fast} takes $O(\sum_{w \in N(v_i)}d(w))$ time.  Steps \ref{alg:C6 fast:N3}~and~\ref{alg:C6 fast:condition} are implemented together, as follows. For Step~\ref{alg:C6 fast:N3}, we traverse each $w_2 \in N(w_1) \cap N_2$, for every $w_1 \in N_1$, and mark every vertex $w_3 \in X(w_2, \sigma(w_1)) \setminus N_1$ with the value $(w_1, w_2)$.  The condition at Step~\ref{alg:C6 fast:condition} is true for $v_i$ if and only if some vertex $w_3$ is marked twice.  Recall that testing membership in $N_1$ and $N_2$ takes $O(1)$ time, while $X(w_2, \sigma(w_1))$ can be obtained in $O(1)$ time while examining $w_2$ in a traversal of $N(w_1)$ with the pointer $x(w_1,w_2)$ (notice that $w_1 \in X(w_2, \sigma(w_1))$ by definition, so $x(w_1, w_2)$ was previously recorded).  As each vertex outside $N_1$ is marked at most twice, the time required by Steps \ref{alg:C6 fast:N3}~and~\ref{alg:C6 fast:condition} is $O(n+\sum_{w_1 \in N_1}d(w_1))$.  Finally, if the condition at Step~\ref{alg:C6 fast:condition} is true, then we obtain a vertex $w_3$ that has two marks, say $(w_1, w_2)$ and $(w_5, w_4)$.  These marks indicate that $w_3 \in X(w_2, \sigma(w_1)) \cap X(w_4, \sigma(w_5))$, this $v_i, w_1, \ldots, w_5$ is a valid output for Step~\ref{alg:C6 fast:halt}.  Clearly, this step takes $O(1)$ time.  Summing up, by~\cite{ChibaNishizekiSJC1985}, Algorithm~\ref{alg:C6 fast} has time complexity
\[
  O\left(\alpha(G)m + \sum_{i = 1}^n\left(\sum \{d(w)+n \mid w \in N(v_i) \cap \MIN(v_i)\}\right)\right) = O(n^2 + \alpha(G)m).
\]

As for the space complexity, recall again that $w$ belongs only to $X(v,\sigma(w))$, for every edge $vw$.  Therefore, as each call to method $C4$ takes $O(m)$ space, Algorithm~\ref{alg:C6 fast} requires $O(n+m)$ bits.

\subsection{Finding an induced $C_5$ efficiently}

An induced $C_5$ can be found in $O(n^2 + m\alpha(G))$ time and linear space, if existing, with a procedure similar to Algorithm~\ref{alg:C6 fast}.  We omit the implementation details, that follow from the next lemma.

\begin{lemma}\label{lem:C5 fast}
  Let $G$ be a degree ordered graph that is $C_4$-dominated and contains no triangles.  Then $G$ contains an induced $C_5$ if and only if it contains a cycle $v_0, \ldots, v_4$  such that $\min\{\sigma(v_2), \sigma(v_3)\} > v_0 > \max\{v_2, v_3\}$.
\end{lemma}

The main theorem of this paper is the following corollary.

\begin{theorem}
  Hereditary biclique-Helly graphs can be recognized in $O(n^2+\alpha(G)m)$ time and $O(n+m)$ space.
\end{theorem}

\section{Maximal bicliques of $C_4$-dominated graphs with no triangles}
\label{sec:bicliques}

In this section we focus on the problem of enumerating all the maximal bicliques of a $C_4$-dominated graph with no triangles.  At the end of this section, we also devote a paragraph to discuss some biclique problems on this class of graphs.

Observe that every pair of twin vertices of any graph $G$ belong to the same maximal bicliques.  Thus, every maximal set of twin vertices can be identified into one vertex so as to obtain a twin-free graph $H$.  Then, for every maximal biclique of $H$, we can replace its vertices with the set that identifies, so as to obtain a maximal biclique of $G$.  So, in this section we assume that $G$ is a degree ordered graph that is $C_4$-dominated and has no triangles nor twins.  

For each $v \in V(G)$, define $B(v) = \{N(v), Dom(v) \cup \{v\}\}$.  Observe that, since $G$ is triangle-free, $B(v)$ is a biclique of $G$.  The following theorem shows that $\{B(v) \mid v \in V(G)\}$ is precisely the family of maximal bicliques of $G$.

\begin{theorem}\label{thm:biclique}
  Let $G$ be a degree ordered graph that is $C_4$-dominated and has no triangles nor twins, and $B$ be a biclique of $G$.  Then, $B$ is a maximal biclique of $G$ if and only if $B = B(v)$ for some $v \in V(G)$.
\end{theorem}

\begin{proof}
  Suppose first that $B$ is a maximal biclique of $G$, and let $\{B_1, B_2\}$ be its bipartition.  Fix $v_1 = \min B_1$ and $v_2 = \min B_2$.  To obtain a contradiction, suppose that neither $B_1 \subseteq Dom(v_1)$ nor $B_2 \subseteq Dom(v_2)$.  Then, there are two vertices $w_1 \in B_1$ and $w_2 \in B_2$ that do not dominate $v_1$ and $v_2$, respectively.  On the other hand, since $v_1 < w_1$ and $v_2 < w_2$, we obtain that $v_1$ and $v_2$ do not dominate $w_1$ and $w_2$, respectively.  But this implies that $v_1, v_2, w_1, w_2$ is a non-dominated $C_4$ of $G$, a contradiction.  Therefore, without loss of generality, we may assume that $B_1 \subseteq Dom(v_1)$.  Now, as $B$ is maximal, it follows that all the vertices that dominate $v_1$ are included in $B_1$ and all the neighbors of $v_1$ are included in $B_2$.  

  The converse is trivial.
\end{proof}

\begin{corollary}
  A $C_4$-dominated graph with no triangles has at most $n$ maximal bicliques.
\end{corollary}

Say that $v \in V(G)$ is a \emph{repeated} vertex if $|Dom(v)| = |N(w)|$ and $|N(v)| = |Dom(w)|$, for some $w \in N(v) \cap \MAX(v)$.  The following lemma shows how to avoid the listing of duplicate maximal bicliques.

\begin{lemma}\label{lem:repeated}
  Let $G$ be a degree ordered graph that is $C_4$-dominated and has no triangles nor twins, and $v \in V(G)$.  Then $B(v) = B(w)$ for some $w \in \MAX(v)$ if and only if $v$ is a repeated vertex.
\end{lemma}

\begin{proof}
  If $B(v) = B(w)$ for some $w \in \MAX(v)$, then, by definition, either $N(v) = N(w)$ or $N(v) = Dom(w)$ and $Dom(v) = N(w)$.  Since $N(v) \neq N(w)$ because $G$ has no twins, the former is impossible.  For the converse, suppose that $v$ is a repeated vertex, i.e., $|Dom(v)| = |N(w)|$ and $|N(v)| = |Dom(w)|$ for some $w \in N(v) \cap \MAX(v)$.  Now, as every vertex in $Dom(v)$ is adjacent to $w$ and every vertex in $Dom(w)$ is adjacent to $v$, it follows that $Dom(v) = N(w)$ and $N(v) = Dom(w)$, thus $B(v) = B(w)$.
\end{proof}

Theorem~\ref{thm:biclique} and Lemma~\ref{lem:repeated} yield a simple $O(nm)$ time and $O(n^2)$ space algorithm for enumerating all the maximal bicliques of $G$.  First, find the set $R$ of repeated vertices by using the domination matrix; then, output $B(v)$ for every $v \in V(G) \setminus R$.  In the remaining of this section, we show an $O(m\alpha(G) + o)$ time and $O(m\alpha(G))$ space algorithm for enumerating all the maximal bicliques of $G$, where $o = \sum_{v \in V(G)}|B(v)|$ is the size of the output.  Our main tool is, in this case, a digraph that encodes all the dominations of the input graph as paths.  This digraph is obtained from the squares domination digraph $S(G)$ by adding all the missing dominations.  

Recall that $S(G)$ is obtained from $U(G)$ by inserting the edge $w \to v$, for every safe $(v,w) \in \MS$.  Say that $w \in V(G)$ is \emph{degenerated} when $w$ has no out-neighbors in $S(G)$ and $N(w) \subseteq \MAX(w)$.  We define the \emph{dominator set} of $w$ according to the following rules.  If $L(S) = \{w\}$ for some $S \in \MS$, then the dominator set of $w$ is empty.  Otherwise, the dominator set of $w$ is $N(z) \setminus \{w\}$, where $z = \min N(w)$.  The following lemma shows that degenerated vertices, together with their dominator sets, are all we need to build the domination digraph from $S(G)$. 

\begin{lemma}\label{lem:low vertex}
  Let $G$ be a degree ordered graph that is $C_4$-dominated and has no triangles nor twins, and $w \in V(G)$ be degenerated. Then, $w$ is dominated by $v \in V(G)$ if and only if $v$ belongs to the dominator set of $w$.  
\end{lemma}

\begin{proof}
  Suppose first that $w$ is dominated by $v \in V(G)$, and let $z = \min N(w)$.  Observe that $v < z$; otherwise, $(v,w)$ would be a safe triple, contradicting the fact that $w$ has no out-neighbors in $S(G)$.  Then, since $v$ dominates $w$, it follows that $v \in L(S)$ for every $S \in \MS$ such that $w \in L(S)$.   Therefore, $L(S) \neq \{w\}$ for every $S \in \MS$, thus the dominator set of $w$ is $N(z) \setminus \{w\}$. Hence, $v$ belongs to the dominator set of $w$.

  For the converse, suppose that $v$ belongs to the dominator set of $w$, and call $z = \min N(w)$.  By definition, $v \in N(z)$.  Note that, if $N(w) = \{z\}$, then $w$ is dominated $v$.  Suppose then that $d(w) > 1$, and take $a \in N(w) \setminus \{z\}$.  Since $N(w) \subseteq \MAX(w)$, we obtain that $a > z > w$, thus $(a, z) \in \MS$.  Also, $L(a,z) \neq \{w\}$, because otherwise the dominator set of $w$ would be empty, and this cannot happen as it contains $v$.  Furthermore, since $w$ has no out-neighbors in $S(G)$ and $|L(a,z)| > 1$, it follows that $(a, z)$ is a safe triple of $\MS$.  Therefore, $a$ dominates $z$, hence $a$ is adjacent to $v$.  Since $a$ is any vertex in $N(w) \setminus \{z\}$, it follows that $v$ is adjacent to all the vertices in $N(w)$, hence $v$ dominates $w$.
\end{proof}

The \emph{domination digraph} of $G$ is the digraph $D(G)$ that is obtained from $S(G)$ by inserting an edge $w \to v$, for every degenerated $w \in V(G)$ and every $v$ in the dominator set of $w$. As its name indicates, $D(G)$ encodes all the dominations in $G$, as the following theorem resumes.

\begin{theorem}\label{thm:domination}
  Let $G$ be a degree ordered graph that is $C_4$-dominated and has no triangles nor twins, and $v,w \in V(G)$.  Then, $v$ dominates $w$ in $G$ if and only if there is a path from $w$ to $v$ in $D(G)$.  Furthermore, if $v$ dominates $w$ and $w \not\to_{D(G)} v$, then there is a path from $w$ to $v$ in $U(G)$.
\end{theorem}

\begin{proof}
  Suppose that $w$ is dominated by $v$.  If $w$ is degenerated, then $w \to v$ is an edge of $D(G)$, by Lemma~\ref{lem:low vertex}.  Otherwise, either $w$ has a neighbor in $\MIN(w)$ or $w$ has an out-neighbor in $S(G)$.  Consider these alternatives:
  \begin{description}
    \item [Alternative 1:] $w$ is adjacent to $z \in \MIN(w)$.  In this case, $z \in L(v,w)$, thus $(v,w)$ is a safe triple of $\MS$.  Therefore, $v$ is an out-neighbor of $w$ in both $S(G)$ and $D(G)$.
    \item [Alternative 2:] $N(w) \subseteq \MAX(w)$ and $d^+_{U(G)}(w) = 0$.  In this case, $w$ has an out-neighbor $z$ in $S(G)$ such that $(z,w)$ is a safe triple of $\MS$.  Then, as $v$ dominates $w$, it follows that $v$ is adjacent to all the vertices in $L(z,w)$.  Therefore, $L(v,w) \neq \emptyset$, thus $(v,w)$ is a safe triple of $\MS$.  Consequently, $w \to v$ is an edge of both $S(G)$ and $D(G)$.
    \item [Alternative 3:] $w$ has an out-neighbor in $U(G)$, say $z$.  For this to happen, there must be an unsafe triple $S \in \MS$ such that $L(S)$ contains both $w$ and $z$.  If $v(S) > v$, then $L(S)$ contains also $v$, and so there is a path from $w$ to $v$ in both $U(G)$ and $D(G)$.  Otherwise, $v(S) < v$ and, since $v$ dominates $w$, it follows that $v(S) \in L(v,w)$.  Hence, $(v,w)$ is a safe triple of $\MS$, and $w \to v$ is an edge of both $S(G)$ and $D(G)$.
  \end{description}

  The converse follows from Lemmas \ref{lem:domination tree}~and~\ref{lem:low vertex} and the definition of $D(G)$, while the furthermore statement follows from the alternatives above. 
\end{proof}

\begin{corollary}\label{cor:domination}
  Let $G$ be a degree ordered graph that is $C_4$-dominated and has no triangles nor twins.  For every $v \in V(G)$, $Dom(v)$ is equal to the set of ancestors of $v$ in $D(G)$.
\end{corollary}

The above corollary can be used to improve the algorithm for listing the maximal bicliques, as it is shown in Algorithm~\ref{alg:bicliques}. The correctness of Algorithm~\ref{alg:bicliques} follows from Lemmas \ref{lem:repeated}~and~\ref{lem:low vertex} and Corollary~\ref{cor:domination}.  With respect to its implementation, the domination matrix is represented simply with successors and predecessors lists, $d(v)$ and $dom(v)$ are stored together with $v$, and $R$ is stored in an vertex with $n$ positions, so that each membership query takes $O(1)$ time.

\begin{algorithm}
  \caption{Maximal bicliques of a $C_4$-dominated graph with no triangles.}\label{alg:bicliques}
  \Input{a degree ordered graph $G$ that is $C_4$-dominated and has no triangles nor twins.}
  \Output{a listing, without duplicates, of all the maximal bicliques in $G$.}
  
  \begin{AlgorithmSteps}
    \Step{Compute the domination matrix $D(G)$.}\label{alg:bicliques:D(G)}
    \Step{For each $v$, set $dom(v)$ as the number of ancestors of $v$ in $D(G)$.}\label{alg:bicliques:dom}
    \Step{Set $R := \{v \in V(G) \mid d(w) = dom(v) \text{ and } d(v) = dom(w) \text{, for some } w \in N(v)\}$.}\label{alg:bicliques:R}
    \Step{For every $v \in V(G) \setminus R$, write $\{Dom(v), N(v)\}$.}\label{alg:bicliques:output}
  \end{AlgorithmSteps}
\end{algorithm}

\subsection{Complexity of Algorithm~\ref{alg:bicliques}}

The following observation is useful in the complexity analysis.

\begin{proposition}\label{prop:edges dg}
  If $G$ is a degree ordered graph that is $C_4$-dominated and has no triangles nor twins, then $D(G)$ has at most $2m$ edges more than $S(G)$.
\end{proposition}

\begin{proof}
  By definition, if $w \to v$ is and edge that belongs to $D(G)$ and does not belong to $S(G)$, then $w$ is a degenerated vertex and $v$ belongs to the dominator set of $w$.  Recall that, by Theorem~\ref{thm:domination}, the dominator set of $w$ is $Dom(w)$, because $w$ has no out-neighbors in $S(G)$.

  Suppose, to obtain a contradiction, that $w_1$ and $w_2$ are degenerated vertices with nonempty dominator sets, and call $z_1 = \min N(w_1)$ and $z_2 = \min N(w_2)$.  If $z_1 = z_2$, then $w_1 \in N(z_2) \setminus w_2 = Dom(w_1)$; analogously, $w_2 \in Dom(w_1)$.  But this contradicts the fact that $w_1$ and $w_2$ are not twins.  Consequently, $\sum\{|Dom(w)| \mid w \in V(G) \text{ and $w$ is degenerated}\} \leq \sum_{z \in V(G)}|N(z)|$, thus $D(G)$ has at most $2m$ edges more than $S(G)$.
\end{proof}

Matrix $D(G)$ can be computed with a five steps procedure, as follows.  First, call Algorithm~\ref{alg:C4 fast} so as to obtain $U(G)$.  Second, traverse each $S \in \MS$ as in Algorithm~\ref{alg:C4 fast} so as to determine whether $S$ is safe or not, and, if $S$ is safe, insert $w(S) \to v(S)$ into $U(G)$ to obtain $S(G)$.  Third, find and mark all the degenerated vertices, by traversing each $w \in V(G)$ while querying whether $w$ has out-neighbors in $S(G)$ and $N(w) \subseteq \MAX(w)$.  Forth, find those degenerated vertices with nonempty $Dom$, by removing the mark to all those vertices that belong to $L(S)$, for some $S \in \MS$ such that $|L(S)| = 1$.  This forth step only requires the traversal of $\MS$.  Finally, $D(G)$ is obtained by inserting an edge $w \to v$ into $S(G)$, for every marked $w$ and every $v \in Dom(w)$.  By Proposition~\ref{prop:edges dg}, $O(m)$ edges are inserted into $S(G)$ in this last step.  So, since the first four steps take $O(\alpha(G)m)$ time, Step~\ref{alg:bicliques:D(G)} of Algorithm~\ref{alg:bicliques} takes $O(\alpha(G)m)$ time.

Step~\ref{alg:bicliques:dom} can be easily implemented so as to run in $O(\alpha(G)m)$ time, by observing that, by Theorem~\ref{thm:domination} $dom(v) = d^+_{D(G)} + \ell(v)$, where $\ell(v)$ is the length of the unique maximal path of $U(G)$ beginning at $v$.  Step~\ref{alg:bicliques:R} clearly takes $O(n+m)$ time.  Finally, Step~\ref{alg:bicliques:output} takes $O(o)$ time, where $o = \sum_{v \in V(G)}{|B(v)|} \leq n^2$ is the size of the output.

As for the space complexity, $D(G)$ is the heaviest structure stored by Algorithm~\ref{alg:bicliques}.  Thus, Algorithm~\ref{alg:bicliques} requires $O(\alpha(G)m)$ space by Proposition~\ref{prop:edges dg}.

\subsection{Some biclique problems}

Theorems \ref{thm:biclique}~and~\ref{thm:domination} can also be used to solve many problems that involve finding a biclique with a certain property.  Just observe that $|Dom(v)|$ can be computed in $O(\alpha(G)m)$ time and $O(n+m)$ space.  Indeed, as in Algorithm~\ref{alg:bicliques}, $|Dom(v)|$ is either the level of $v$ in $U(G)$ plus the number of safe triples $(w,v) \in \MS$ with $v \not\to_{U(G)} w$ (when $v$ is not degenerate) or the number of vertices in its dominator set (when $v$ is degenerate). (Twin vertices are ignored by Theorems \ref{thm:biclique}~and~\ref{thm:domination}; nevertheless, they can be counted as well when $|Dom(v)|$ is being computed.) Therefore, for $k \in \mathbb{N}$, we can test, in $O(\alpha(G)m)$ time and $O(n+m)$ space, whether $G$ has a biclique $(B_1, B_2)$ such that: a) $|B_1| = |B_2| = k$, b) $|B_1| + |B_2| \geq k$, or $|B_1| \cdot |B_2| \geq k$.  These problems are respectively called: a) balanced biclique problem, b) maximum vertex biclique problem, and c) maximum edge biclique problem, and are all NP-complete for general graphs~\cite{Garey1979,PeetersDAM2003}.  Furthermore, problem c) is NP-complete even when $G$ is a bipartite graph~\cite{PeetersDAM2003}.

\section{Concluding remarks and open problems}
\label{sec:remarks}

In this paper we devised an efficient algorithm for the recognition of hereditary biclique-Helly graphs.  The algorithm is strongly dependent on the $C4$ method developed by Chiba and Nishizeki in~\cite{ChibaNishizekiSJC1985}.  The techniques developed by these authors were extended recently by Lin et al.\ in~\cite{LinSoulignacSzwarcfiter2010}.  In the latter paper, these techniques are used to recognize some graph classes that are defined in terms of the true domination relation, such as cop-win graphs and strongly chordal graphs.  A vertex $v$ is \emph{true dominated} by a vertex $w$ when $N(v) \cup \{v\} \subseteq N(w) \cup \{w\}$.  One of the appealing aspects of our algorithm, is that it has to deal with the domination relation.  

In Section~\ref{sec:C6 fast}, we showed how to find an induced $C_6$ in a $C_4$-dominated graph with no triangles.  As part of Algorithm~\ref{alg:C6 fast}, we had to find some vertices that appear at distance at most $3$ from a given vertex $v$.  We showed how to do this in $O(n)$ time, by traversing each vertex outside $N(v)$ at most twice.  However, not all the vertices outside $N(v)$ are at distance $3$ from $G$, unless $G$ is somehow dense.  An open question that follows is, then, can an induced $C_6$ be found in $O(\alpha(G)m)$ time?  Furthermore, is it possible to do find such $C_6$ by using only $O(m)$ space?

In Section~\ref{sec:bicliques}, we develop an algorithm for enumerating the bicliques in $O(o+\alpha(G)m)$ time and $O(\alpha(G) m)$ space.  Algorithm~\ref{alg:bicliques} uses the domination digraph $D(G)$ that encodes all the dominations in $O(\alpha(G)m)$ space.  We can divide the dominations encoded in $D(G)$ into three classes.  Those corresponding to paths in $U(G)$; those corresponding to edges in $S(G) - U(G)$; and those corresponding to edges in $D(G) - S(G)$.  Although there can be as many as $O(\sum_{v \in V(G)}Dom(v)) = O(o)$ dominations of the first class, only $O(n)$ space is used by $D(G)$ to encode this information.  On the other hand, Proposition~\ref{prop:edges dg} guaranties that there are only $O(m)$ dominations of the third class.  Finally, $O(\alpha(G)m)$ bits are used to store the $O(\alpha(G)m)$ dominations second class.  Could it be possible to encode such dominations in a different way (e.g.~as paths), so that only $O(m)$ space is used, without affecting the time complexity of the algorithm?

The somehow opposite question is also interesting.  Instead of implicitly storing each domination of the first class as a path, we can store it explicitly.  That is, we can insert an edge $w \to v$ into $D(G)$ for every path of $U(G)$ beginning at $w$ and ending at $v$.  Such digraph would require $O(o)$ space.  We know that $o < n^2$; however, for $w$ to be dominated by a vertex at distance $k$ from $v$ in $D(G)$, there should be at least $k$ vertices in $N(w) \setminus N(v)$ (always assuming that $G$ is twin-free).  Thus if $D(G)$ has many large paths, then $G$ is dense and $\alpha(G)$ is large.  So, it would be nice to prove or disprove that $o \in O(\alpha(G)m)$, or that $o \in O(m^{3/2})$.

Let $\mathcal{B}$ be the class of graph formed by those graphs $G$ such that $\{B(v)\}_{v \in V(G)}$ is the family of maximal bicliques of $G$.  By Theorem~\ref{thm:biclique}, every $C_4$-dominated graph with no triangles belongs to $\mathcal{B}$.  The converse is false, since the graph in Figure~\ref{fig:ladder with vertex} belongs to $\mathcal{B}$ and it contains a non-dominated $C_4$.  The graphs in $\mathcal{B}$ have at most $n$ bicliques, and Lemma~\ref{lem:repeated} holds for all the graphs in $\mathcal{B}$.  Hence, the simple algorithm that lists $\{B(v)\}_{v \in V(G)}$ by using the domination matrix, implies that the maximal bicliques of the graphs in $\mathcal{B}$ can be enumerated in $O(nm)$ time and $O(n^2)$ space.  Is it possible to list the maximal bicliques of the graphs in $\mathcal{B}$ faster? Is it possible obtain the same time and space bounds of Algorithm~\ref{alg:bicliques} for a larger class of graphs?

\begin{figure}
  \centering \includegraphics{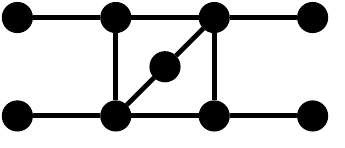}
  \caption{A graph $G$ that contains a non dominated $C_4$, whose family of maximal bicliques is $\{B(v)\}_{v \in V(G)}$.}\label{fig:ladder with vertex}
\end{figure}


\end{document}